\documentclass[sigconf,authorversion]{acmart} 

\usepackage{amssymb}
\usepackage{amsthm}
\usepackage{array}
\usepackage{algorithm}
\usepackage{algorithmic}
\usepackage{balance}  

\AtBeginDocument{%
  \providecommand\BibTeX{{%
    \normalfont B\kern-0.5em{\scshape i\kern-0.25em b}\kern-0.8em\TeX}}}

\newcolumntype{P}[1]{>{\centering\arraybackslash}p{#1}}

\newtheorem{theorem}{Theorem}[section]

\copyrightyear{2026}
\acmYear{2026}
\setcopyright{cc}
\setcctype{by}
\acmConference[KDD '26]{Proceedings of the 32nd ACM SIGKDD Conference on Knowledge Discovery and Data Mining V.2}{August 09--13, 2026}{Jeju Island, Republic of Korea}
\acmBooktitle{Proceedings of the 32nd ACM SIGKDD Conference on Knowledge Discovery and Data Mining V.2 (KDD '26), August 09--13, 2026, Jeju Island, Republic of Korea}
\acmDOI{10.1145/3770855.3828560}
\acmISBN{979-8-4007-2259-2/2026/08}

\begin{document}

\title{Scalable Network-Aware Experiment Design for Two-Sided Marketplaces}

\author{Yi Su}
\authornote{Both authors contributed equally to this research.}
\affiliation{%
  \institution{LinkedIn Corporation}
  \streetaddress{800 E Middlefield Rd}
  \city{Mountain View}
  \state{CA}
  \country{USA}}
\email{ysu1@linkedin.com}
\orcid{0000-0002-4778-1161}

\author{Zhen Yan}
\authornotemark[1]
\affiliation{%
  \institution{LinkedIn Corporation}
  \streetaddress{800 E Middlefield Rd}
  \city{Mountain View}
  \state{CA}
  \country{USA}}
\email{zhyan@linkedin.com}
\orcid{0009-0000-7386-537X}

\renewcommand{\shortauthors}{Yi Su and Zhen Yan}

\begin{abstract}
Measuring causal effects in networked two-sided marketplaces is challenging due to treatment interference between market participants on different sides. When treatment is applied to one side (e.g., job seekers), their interactions with the other side (e.g., job posters) introduce spillover effects that violate the Stable Unit Treatment Value Assumption (SUTVA) and bias causal estimates. While cluster-based randomization mitigates this problem, prior approaches struggle with a fundamental trade-off: reducing spillover requires isolated clusters that will reduce the number of qualifying clusters, which decreases statistical power. This paper introduces \textit{EgoCluster V3}, an iterative clustering algorithm that reduces spillover by 3$\times$ compared to prior versions while preserving node coverage and doubling test power. We further introduce \textit{MultiEgoCluster}, which extends V3 through a two-stage procedure that first groups highly connected egos into multi-ego clusters before applying the iterative clustering algorithm. This achieves an additional $\sim$56\% spillover reduction and $\sim$38\% increase in sample size. Both methods are deployed in production at LinkedIn and have systematically enabled high-impact two-sided marketplace experiments. Since residual bias cannot be fully eliminated through clustering alone, we derive a theoretical bias correction method for average treatment effect (ATE) estimation based on graph structure and propose an approach to generalize results to the general population.
\end{abstract}




\maketitle

\section{Introduction}
\label{sec:intro}

Randomized controlled experiments are the gold standard for causal inference in industry and research. However, their applicability is inherently limited in networked environments where treatment effects spill over between connected units. This violation of SUTVA (Stable Unit Treatment Value Assumption) is particularly acute in two-sided marketplaces such as LinkedIn's feed (viewer vs creator) and job marketplaces (poster vs seeker), where actions by one participant directly affect outcomes for connected participants. For example, a job seeker viewing a job poster's posting would affect the job poster's engagement.

Consider measuring a relevance model's impact on creators in a typical feed marketplace: when we test a new feed ranking algorithm, viewers are randomized into treatment or control group. However, creators' behavior (our measurement target) depends on the aggregated feedback from both treatment and control viewers. If a treated viewer sends more positive feedback than a control viewer, this will move the receiving creator's metrics even without any direct intervention on the creators themselves. This spillover invalidates traditional ATE estimates and makes it impossible to measure the creator-side impact with sufficient statistical power.

Cluster-based randomization addresses this by grouping connected units into clusters and randomizing entire clusters rather than individual units. The intuition is straightforward: if we keep viewers (we call them ``alters'' in the EgoCluster context) who interact with the same creator (ego) together in the same treatment group, we isolate the creator from mixed treatment exposure. EgoCluster, first introduced in 2019 to measure network effects at LinkedIn \cite{Guillaume2019}, pioneered this approach. EgoCluster V2, released in 2024, improved scalability and increased sample size by 5 times through a one-degree label propagation algorithm that efficiently assigns alters to egos based on their strongest connections \cite{wesu2024improving}.

The original EgoCluster is a greedy algorithm that controls spillover well, but it suffered from scalability issue and smaller sample sizes hence smaller power, which EgoCluster V2 aims to address. However, V2 faced a critical limitation: despite generating many more clusters than V1, spillover remained high ($\approx$20\% loss rate at cluster creation). Naive approaches to control spillover---such as discarding all high-loss clusters---remove too many users and provide little spillover benefit. Another challenge is that V2 lacks a method to generalize the observed impact from an EgoCluster experiment to the general population. This paper addresses these limitations through four core innovations:

\paragraph{1. EgoCluster V3.} An iterative clustering algorithm that gradually removes only the ``worst'' egos and re-assigns their alters to remaining egos. This iterative approach reduces spillover from 30\% to <10\% while maintaining high node coverage, improving statistical power by 2$\times$.

\paragraph{2. MultiEgoCluster.} A two-stage procedure that first groups similar egos as a multi-ego using label propagation before conducting iterative clustering as in V3. This stage prevents highly connected egos from being split across treatment variants, and it leads to an additional $\sim$56\% spillover reduction and increased sample sizes ($\sim$38\%).

\paragraph{3. Spillover Bias Correction.} A theoretical framework to correct ATE estimates for residual interference based on graph structure assumptions. Since clustering cannot eliminate all spillover, we derive the bias as a function of the observable network topology.

\paragraph{4. Site-Wide Impact (SWI) Estimation.} Because of the spillover constraints, EgoCluster algorithms tend to create an unrepresentative sample of creators or egos. We develop a method using Conditional Average Treatment Effects (CATEs) with Coarsened Exact Matching (CEM) weighting to generalize the treatment effect back to the full creator population to obtain a SWI estimate.

~

These methods have been deployed in production at LinkedIn and have enabled measurement of two-sided effects that would otherwise go undetected due to low power.

\section{Related Work}
\label{sec:related}

\subsection{Network Interference and SUTVA Violations}

The problem of interference in networked systems has received substantial attention in the causal inference literature. Fisher's foundational work on randomized controlled experiments \cite{fisher1935design} implicitly assumed units were independent, but violations of SUTVA are ubiquitous in social networks and marketplaces. Rubin's potential outcomes framework \cite{rubin1974estimating} formalizes SUTVA as the requirement that a unit's outcome depends only on its own treatment assignment, not on others' assignments. Holland \cite{holland1986statistics} provided a comprehensive statistical framework for causal inference that highlighted the importance of these assumptions.

Violations manifest in two primary forms: social spillovers (where behavior influences connected peers) and competitive spillovers (where interventions affect market competitors). Manski \cite{manski1993identification} formalized the ``reflection problem'' in identifying social effects, which is closely related to interference in experiments. Empirical evidence of SUTVA violations has been documented across platforms \cite{bakshy2012design}. Facebook reported substantive interference in content ranking experiments through both direct network effects and competition effects. LinkedIn observed similar patterns, with creator feedback from viewers influencing creator behavior \cite{xu2015infrastructure}.

\subsection{Cluster-Based Randomization}

Cluster-based randomization is an established design-based approach to mitigating network interference \cite{box2005, tamhane2000statistical}. Instead of randomizing individuals, treatment is assigned at the cluster level, with all units in a cluster receiving the same treatment. The efficacy depends critically on cluster structure: tight within-cluster connections and sparse between-cluster edges minimize spillover.

Ugander et al. \cite{ugander2013graph} introduced graph cluster randomization and analyzed network exposure to multiple treatment universes, providing theoretical foundations for cluster-based experiments. Meta's network experimentation framework demonstrated the practical value of clustering \cite{karrer2021network}. Their approach uses community detection algorithms on the social graph to form clusters, then conducts cluster-level randomization. Subsequent theoretical work by Aronow and Samii \cite{aronow2017estimating} established that when interference is limited to within-cluster neighbors (partial interference), cluster randomization can enable unbiased ATE estimation.

The EgoCluster framework was first introduced by Saint-Jacques et al. \cite{Guillaume2019} for measuring network effects at LinkedIn. The V2 iteration \cite{wesu2024improving} improved upon this by introducing a one-degree label propagation algorithm that increased sample size five-fold while maintaining clustering quality.

\subsection{ATE Estimation Under Residual Spillover}

While cluster randomization reduces spillover, residual interference typically remains. When between-cluster spillover exists, naive ATE estimators are biased. Eckles et al. \cite{eckles2017design} provided design and analysis frameworks for experiments in networks that explicitly account for interference. Inverse Probability Weighting (IPW) and Horvitz-Thompson estimators weight observations by the inverse probability of exposure \cite{imbens2015causal}. Forastiere et al. \cite{forastiere2021identification} provided identification and estimation strategies for treatment effects in the presence of general interference patterns. Design-based randomization provides conservative variance estimates without assuming specific functional forms. Saveski et al. \cite{saveski2017detecting} introduced methods for detecting network effects by randomizing over randomized experiments. Our approach combines design-based thinking with structural assumptions about interference patterns specific to ego-cluster experiments.

\subsection{Two-Sided Marketplace Experiments}

Two-sided marketplaces introduce unique experimental challenges. Beyond network spillovers, competitive effects arise when interventions on supply-side units affect demand-side units through changed equilibrium \cite{weintraub2022experimentation}. Airbnb's work on cluster randomization \cite{holtz2024airbnb} demonstrated that cluster-based approaches can reduce interference bias by 20\% or more, though focused on competitive rather than network spillovers.
Large-scale experimentation infrastructure has been developed at major tech companies to handle these challenges \cite{tang2010overlapping, kohavi2009controlled, kohavi2017online}. Practical guidance on running such experiments at scale has been documented \cite{hohnhold2015focusing, deng2013improving, kohavi2014seven}. Chin \cite{chin2019controlling} provided regression-based methods for controlling for network effects in experiments that complement clustering approaches.

\subsection{Generalization from Non-Representative Samples}

When clustering creates sample attrition, the remaining sample may not represent the target population. Epasto et al. \cite{epasto2017ego} developed ego-splitting frameworks for handling non-overlapping to overlapping cluster scenarios. Standard approaches to generalization include propensity score stratification and inverse probability weighting, which we employ for site-wide effect estimation through Coarsened Exact Matching (CEM) \cite{imbens2015causal}.

\section{Problem Formulation and Background}
\label{sec:problem}

\subsection{Experimental Setting}

We work with a bipartite network $G = (E \cup A, W)$ where $E$ represents egos (one side of the marketplace, e.g., creators) and $A$ represents alters (the other side, e.g., viewers). Weighted edges $w_{ij}$ capture the strength of past interactions (e.g., number of feed impressions). Our goal is to assign egos to treatment ($T$) or control ($C$) groups such that: (1) spillover is minimized, (2) coverage is preserved, and (3) power is maximized.

Once egos are assigned, alters inherit their ego's treatment assignment based on the clustering algorithm. During the experiment, egos remain in control while alters receive treatment/control variants. The treatment effect is measured on egos.

\subsection{EgoCluster V2 Baseline}

EgoCluster V2 \cite{wesu2024improving} assigns alters to egos based on strongest connections. For each alter $j$, it computes aggregated interaction weights to treatment and control egos, then assigns the alter to the variant with higher aggregated weight:
\begin{equation}
\text{variant}(j) = \arg\max_{v \in \{T, C\}} \sum_{i \in v} w_{ij}.
\end{equation}
This approach generates 5$\times$ more clusters than V1 but leaves spillover high ($\approx$20--30\% loss rate at cluster creation). The loss rate is defined as the fraction of an ego's total alter weight that is assigned to the opposite variant:
\begin{equation}
\label{eq:loss_rate}
\text{LR}_i := \frac{\sum_{j: \text{variant}(j) \neq \text{variant}(i)} w_{ij}}{\sum_{j} w_{ij}}.
\end{equation}
The overall loss rate (OLR) aggregates individual ego loss rates weighted by their total connection weights.

\subsection{Key Challenge: Why Simple Filtering Failed}
\label{sec:filtering_failed}

A natural approach to reduce spillover is to discard egos with high loss rates. We investigated this in a preliminary study called ``EgoCluster V2.5,'' which filtered egos with loss rate exceeding a threshold (e.g., 20\%).
However, this approach fell short for two reasons:

\paragraph{1. Excessive Sample Loss.} When filtering egos with loss rates above 20\%, we removed approximately 85\% of all egos---from 2,209,225 egos down to only 338,577 egos. This drastic reduction in experimental sample size severely hurts statistical power.

\paragraph{2. Minimal Spillover Improvement.} Despite removing 85\% of egos, the Overall Loss Rate (OLR) only decreased from 30.03\% to 28.73\%---a mere 1.3 percentage point improvement. This is because alters of removed egos remain connected to the remaining egos, and many of these connections cross variant boundaries.

~

The underlying issue is that high-loss egos often have large networks with many shared alters. Simply removing these egos does not reassign their alters; it just excludes them from the experiment entirely. The remaining egos still experience spillover from these orphaned alters who interact with egos in both variants.

This observation motivated our iterative approach: instead of discarding high-loss egos and their alters entirely, we remove egos but \textit{re-assign} their alters to the remaining egos, gradually improving the overall clustering quality.

\section{EgoCluster V3: Iterative Clustering Algorithm}
\label{sec:v3_algorithm}

\subsection{Algorithm Overview}

EgoCluster V3 addresses the limitations of simple filtering through an iterative refinement process. The key insight is that by removing high-loss egos and \textit{re-assigning} their alters to remaining egos using the V2 clustering logic, we can progressively reduce spillover while maintaining alter coverage.

The algorithm proceeds through the following steps:

\paragraph{Step 1: Data Loading and Preprocessing.} Load the complete bipartite graph $G = (E \cup A, W)$ where edges represent ego-alter interactions weighted by engagement (e.g., impressions). Compute basic graph statistics for later evaluation.

\paragraph{Step 2: Network Size Filtering.} Remove extremely large egos exceeding the $p$-th percentile of network size (typically $p = 99.9$). These ``super-influencers'' have disproportionately many alters and generate excessive spillover that cannot be mitigated through clustering alone.

\paragraph{Step 3: Random Ego Assignment.} Randomly assign remaining egos to treatment or control variants with probability 0.5 each. Critically, this assignment is performed \textbf{once} at the beginning and persisted using checkpointing to prevent data loss during iterative processing.

\paragraph{Step 4: Initial Clustering (Iteration 0).} Apply the EgoCluster V2 clustering algorithm:
\begin{itemize}
    \item For each alter $j$, compute aggregated weights to treatment egos ($W_T(j)$) and control egos ($W_C(j)$)
    \item Assign $j$ to variant $\arg\max(W_T(j), W_C(j))$
    \item Within the assigned variant, assign $j$ to the ego with the highest individual weight
\end{itemize}
This per-alter argmax rule is optimal given a fixed set of active egos
(Theorem~\ref{thm:alter_optimality}), which is the property V3 exploits:
each iteration re-applies it to a shrinking, higher-quality ego set, so
the per-iteration clustering remains optimal conditional on the surviving egos.

\paragraph{Step 5: Loss Computation.} Compute two metrics for each ego:
\begin{itemize}
    \item \textbf{Loss Rate ($\text{LR}_i$)}: Fraction of ego $i$'s total weight connected to alters assigned to the opposite variant
    \item \textbf{Loss Weights}: Absolute weight of misaligned connections
\end{itemize}
Also compute the overall spillover summary across all egos.

\paragraph{Step 6: Iterative Refinement.} For each iteration $k = 1, 2, \ldots, K$:
\begin{enumerate}
    \item Load the ego-level loss metrics from the previous iteration
    \item Identify the ``worst'' egos to remove based on loss criteria (see Section~\ref{sec:stratified_filtering})
    \item Remove these egos from the active set
    \item Re-run the V2 clustering on the remaining egos and all alters
    \item Recompute loss rates and write results
\end{enumerate}

\paragraph{Step 7: Output Generation.} After all iterations complete, generate experiment assignment files for the desired iteration (typically the one achieving target spillover with acceptable coverage).

\subsection{Stratified Filtering with Features}
\label{sec:stratified_filtering}

A naive approach would rank all egos globally by loss rate and remove the top fraction. However, this can disproportionately remove egos from certain subpopulations, creating imbalanced experiments.

EgoCluster V3 implements stratified filtering based on ego features that capture platform tenure, engagement level, and activity patterns. Within each stratum defined by the cross-product of variant and feature buckets, we independently identify the top $s\%$ of egos by loss (where $s$ is the step size, typically 10\%). This ensures balanced filtering across ego subpopulations.

The filtering criterion can be either Loss Weights (removes egos contributing the most absolute spillover) or Loss Rate (removes egos with the highest proportional spillover). Empirically, using loss weights performs well as it prioritizes high-impact egos.

\subsection{Implementation Details}

EgoCluster V3 is implemented in Apache Spark using DataFrames for distributed computation. Key implementation aspects include:

\paragraph{Checkpointing.} The initial ego assignment DataFrame is checkpointed (not just cached) to prevent data loss under Spark's LRU eviction policy. This ensures consistent treatment assignment across all iterations.

\paragraph{Iteration Output.} Results are written after each iteration, partitioned by version (iter0, iter1, ...), enabling comparison across iterations without recomputation.

\subsection{Complexity Analysis}

Let $|E|$ be the number of egos, $|A|$ be the number of alters, and $|W|$ be the number of edges.

\begin{itemize}
    \item \textbf{Initial clustering}:
    \begin{itemize}
        \item $O(|W|)$ for aggregating weights
        \item $O(|A|\log|A|)$ for identifying the maximum per alter
    \end{itemize}
    \item \textbf{Loss computation}: $O(|W|)$ for joining cluster assignments with full graph
    \item \textbf{Per-iteration filtering}: $O(|E| \log |E|)$ for ranking and filtering
    \item \textbf{Total}: $O(K \cdot (|W| + |E| \log |E|))$ for $K$ iterations
\end{itemize}

For production networks with $\sim$1.6M egos and $\sim$100M edges, the full pipeline executes in approximately 2 hours on a Spark cluster.

\subsection{Evaluation Metrics}
\label{sec:eval_metrics}

We evaluate clustering quality using multiple metrics:

\paragraph{1. Loss Rate at T0 (Cluster Creation Time).} The overall loss rate computed immediately after clustering, using the same graph used for clustering. This measures how well the algorithm minimizes spillover given the input network structure.

\begin{equation}
\text{OLR}_{T0} = \frac{\sum_{i \in E} \text{loss\_weights}_i}{\sum_{i \in E} \text{total\_weights}_i}.
\end{equation}

\paragraph{2. Loss Rate at T14 (14 Days Post-Clustering).} The loss rate computed using a graph from 14 days after cluster creation. This measures \textit{stability}---how well the clustering holds up as user behavior evolves. We define two scenarios:
\begin{itemize}
    \item \textbf{Scenario 1}: Egos keep their assigned alters (cluster\_alters) from T0, and we measure their loss against the T14 graph
    \item \textbf{Scenario 2}: Each ego's cluster expands to include all alters in their network at T14 (network\_alters), and we measure loss for this expanded set
\end{itemize}

Scenario 2 is more realistic as it captures new connections formed during the experiment.

\paragraph{3. Stability Rate.} The fraction of ego-alter edges at T0 that still exist at T14:
\begin{equation}
\text{Stability} = \frac{|W_{T0} \cap W_{T14}|}{|W_{T0}|}.
\end{equation}
High stability (>80\%) validates the assumption that network structure remains relatively constant during the experiment.

\paragraph{4. Node Coverage.} The fraction of original egos (and alters) retained after iterative filtering:
\begin{equation}
\text{Coverage} = \frac{|E_{\text{final}}|}{|E_{\text{initial}}|}.
\end{equation}

\paragraph{5. Minimum Detectable Effect (MDE).} The smallest effect size detectable at 80\% power and 5\% significance level. MDE is inversely proportional to $\sqrt{n}$ and benefits from both increased sample size and reduced variance (from lower spillover).

\section{MultiEgoCluster: Two-Stage Clustering}
\label{sec:multiego}

While EgoCluster V3 achieves significant spillover reduction, a complementary innovation addresses a different source of spillover: egos that share many alters but are assigned to different variants.

\subsection{Motivation}

Consider two job posters in the same industry who receive applications from overlapping sets of job seekers. If these egos are assigned to different variants, their shared alters create spillover regardless of how the clustering algorithm assigns individual alters.

MultiEgoCluster addresses this by grouping similar egos into ``multi-egos'' treated as single units for randomization. It is most effective when the ego-ego graph (induced by shared alters) has a clear community structure---groups of egos sharing many alters internally but few across---as in job marketplaces, where posters in the same industry or region share seekers. When this structure is absent, inter-ego coupling is weak and indirect spillover paths are limited, so EgoCluster V3 alone suffices: MultiEgoCluster's marginal benefit is smallest precisely where it is needed least.

\subsection{Stage 1: Ego Grouping via Label Propagation}

We construct an ego-ego similarity graph and apply label propagation to form groups:

\paragraph{Step 1: Build Ego-Ego Graph.} For each pair of egos $(i, j)$, if they have shared alters, create an edge between them. The edge weight equals the number of shared alters:
\begin{equation}
\text{weight}(i, j) = |A_i \cap A_j|.
\end{equation}
where $A_i$ is the set of alters connected to ego $i$. Two egos are considered connected if they have at least one shared alter.

\paragraph{Step 2: Label Propagation.} Apply the label propagation algorithm \cite{raghavan2007near} to partition egos into communities. Each ego starts with a unique label, then iteratively adopts the most common label among its neighbors until convergence.

\paragraph{Step 3: Multi-Ego Formation.} Each connected component (community) forms a multi-ego. All egos in a multi-ego will receive the same treatment assignment.

\subsection{Stage 2: Iterative Clustering on Multi-Egos}

After forming multi-egos, we apply EgoCluster V3 with multi-egos as the randomization units:

\begin{enumerate}
    \item Randomly assign each multi-ego to treatment or control.
    \item For alter assignment, aggregate multi-ego weights $W_M(j) = \sum_{i \in M} w_{ij}$ for multi-ego $M$.
    \item Assign alters to the multi-ego (hence variant) with highest aggregated weight.
    \item Proceed with iterative refinement as in V3, but removing multi-egos rather than individual egos.
\end{enumerate}

\subsection{Theoretical Justification}

Grouping similar egos reduces spillover by ensuring that shared alters do not cross variant boundaries. If egos $i$ and $j$ share many alters and are in the same multi-ego, all their shared alters will be assigned to the same variant, eliminating this source of spillover.
The trade-off is reduced granularity: with fewer randomization units (multi-egos vs. individual egos), variance per unit may increase. However, the substantial spillover reduction typically outweighs this effect.

\subsection{Algorithm Framework}

\begin{figure}[h]
\centering
\includegraphics[width=0.95\linewidth]{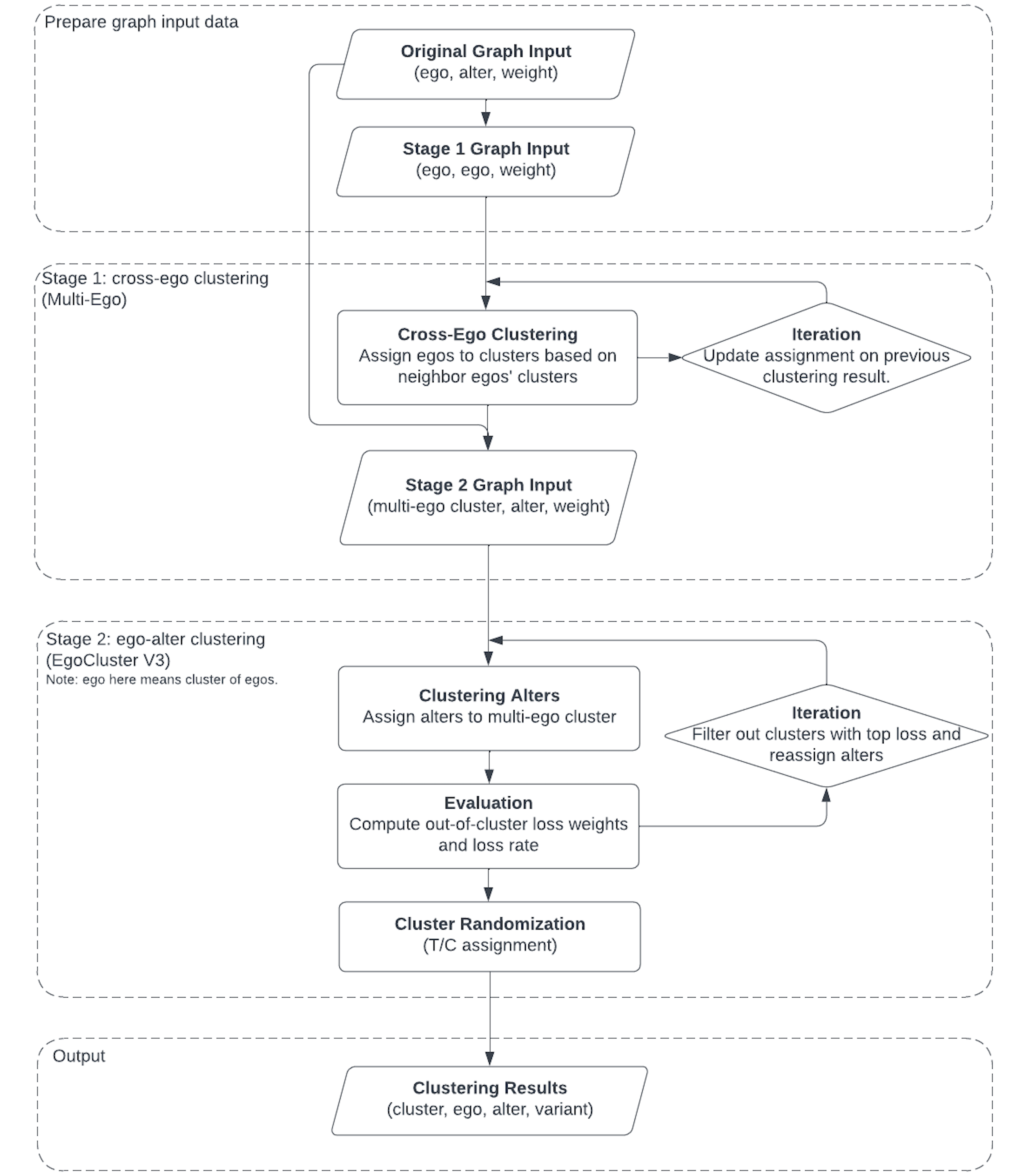}
\caption{\textbf{MultiEgoCluster Algorithm Overview.}
MultiEgoCluster combines two stages: 1) cross-ego clustering to form multi-ego clusters, 2) EgoCluster V3 for ego-alter clustering.}
\Description{Two-stage flowchart: Stage 1 groups similar egos using a shared-alter graph and label propagation, producing multi-ego clusters; Stage 2 applies EgoCluster V3 iterative clustering to those multi-egos and assigns alters.}
\label{fig:algo}
\end{figure}

\section{Bias Correction for Residual Spillover}
\label{sec:bias_correction}

Even with optimized clustering, residual spillover remains. In this section, we derive a theoretical framework for correcting the observed ATE and generalizing the impact to site-wide.

\subsection{Notation and Setup}

For each ego $i$, let:
\begin{itemize}
    \item $Y_i$: Observed outcome for ego $i$
    \item $n_i^T$: Number of treatment-assigned alters connected to ego $i$
    \item $n_i^C$: Number of control-assigned alters connected to ego $i$
    \item $n_i = n_i^T + n_i^C$: Total alters connected to ego $i$
    \item $\tau$: True treatment effect per alter
\end{itemize}

\subsection{Potential Outcomes Framework}

Under the potential outcomes framework with one-hop interference, ego $i$'s outcome depends on the number of treated vs. control alters it receives:
\begin{equation}
Y_i = Y_i(0) + \tau \cdot n_i^T.
\end{equation}
where $Y_i(0)$ is the baseline outcome with no treated alters, and $\tau$ captures the marginal effect of each treated alter.

For a treatment-assigned ego ($i \in T$), ideally all alters would be treated ($n_i^T = n_i$), yielding $Y_i = Y_i(0) + \tau \cdot n_i$. However, due to spillover, some alters may be assigned to control, so $n_i^T < n_i$.

\subsection{ATE Derivation}

The naive ATE estimator compares mean outcomes between treatment and control egos:
\begin{equation}
\hat{\mu}_Y^{\text{naive}} = \frac{1}{|T|} \sum_{i \in T} Y_i - \frac{1}{|C|} \sum_{i \in C} Y_i.
\end{equation}

Let $\alpha_T = \mathbb{E}[n_i^T / n_i | i \in T]$ be the average fraction of treatment alters for treatment egos (ideally 1, but less due to spillover), and $\alpha_C = \mathbb{E}[n_i^T / n_i | i \in C]$ be the same for control egos (ideally 0, but positive due to spillover).
Then, the expected naive ATE is:
\begin{equation}
\mathbb{E}[\hat{\mu}_Y^{\text{naive}}] = \tau \cdot \bar{n} \cdot (\alpha_T - \alpha_C).
\end{equation}
where $\bar{n}$ is the average number of alters per ego. Note that $(\alpha_T - \alpha_C) < 1$ due to spillover, so the naive estimator is attenuated.

\subsection{Bias Correction Formula}

The corrected ATE estimator accounts for the incomplete treatment separation:
\begin{equation}
\hat{\tau}^{\text{corrected}} = \frac{\hat{\mu}_Y^{\text{naive}}}{\bar{n} \cdot (\hat{\alpha}_T - \hat{\alpha}_C)}.
\end{equation}
where $\hat{\alpha}_T$ and $\hat{\alpha}_C$ are estimated from the clustering assignment:
\begin{align}
\hat{\alpha}_T &= 1 - \text{OLR}_T, \\
\hat{\alpha}_C &= \text{OLR}_C.
\end{align}

Here, $\text{OLR}_T$ is the overall loss rate for treatment egos (fraction of their alter weight in control), and $\text{OLR}_C$ is the loss rate for control egos (fraction in treatment).

\paragraph{Interpretation.} The correction inflates the observed effect by the factor $1/(\hat{\alpha}_T - \hat{\alpha}_C)$. When spillover is low, $\hat{\alpha}_T \approx 1$ and $\hat{\alpha}_C \approx 0$, so the correction factor approaches 1. With high spillover (e.g., $\hat{\alpha}_T = 0.8$, $\hat{\alpha}_C = 0.2$), the correction factor is $1/0.6 \approx 1.67$.

\subsection{Site-Wide Impact (SWI) Estimation}
\label{sec:swi}

Besides the unaccounted spillover effects, another challenge for reporting the site-wide impact is underrepresentativity introduced by the clustering procedure.
Because filtering removes certain types of egos (particularly those with high network connectivity), the remaining ego sample may not be representative of the full creator population. To estimate the site-wide impact, we need to reweight the experimental sample to match the population.

We employ Conditional Average Treatment Effects (CATEs) with Coarsened Exact Matching (CEM) weighting:

\paragraph{Step 1: Define Strata.} Coarsen continuous covariates into buckets based on ego characteristics such as platform tenure, engagement level, and topology-aware covariates e.g., network size.

\paragraph{Step 2: Compute CATEs.} For each stratum $s$ defined by the cross-product of bucket values, compute the conditional ATE:
\begin{equation}
\hat{\tau}_s = \bar{Y}_s^T - \bar{Y}_s^C.
\end{equation}
where $\bar{Y}_s^T$ and $\bar{Y}_s^C$ are the mean outcomes for treatment and control egos in stratum $s$.

\paragraph{Step 3: Compute Population Weights.} For each stratum $s$, compute the population weight:
\begin{equation}
w_s = \frac{N_s^{\text{pop}}}{\sum_{s'} N_{s'}^{\text{pop}}}.
\end{equation}
where $N_s^{\text{pop}}$ is the count of egos in stratum $s$ in the full population (including filtered egos).

\paragraph{Step 4: Compute Weighted SWI.} The site-wide impact is the weighted average of CATEs:
\begin{equation}
\hat{\tau}^{\text{SWI}} = \sum_s w_s \cdot \hat{\tau}_s.
\end{equation}

This approach ensures that the estimated effect reflects the composition of the full population, not just the experimental sample. Strata with many filtered egos (high $w_s$ but few observations) receive appropriate upweighting.

\paragraph{Trimming (Optional).} For strata with very few observations (insufficient for reliable CATE estimation), we can trim by excluding them from the weighted average and adjusting weights accordingly.

\paragraph{Discussion on uncertainty quantification.} Since the CEM is based on stratification features that capture heterogeneity, we can ablate one feature at a time and recompute the weights. Large movements indicate the dropped feature carries more heterogeneity and needs to be kept or refined.

The standard confidence interval of $\hat{\tau}^{\text{SWI}}$ can be estimated from the standard error of $\hat{\tau}_s$ for each stratum when $|T_s|$ and $|T_c|$ are sufficient. For long-tail stratification features, we may end up with strata that have few units, where we can estimate by bootstrapping with replacement. In practice, we derive heuristic lower and upper bounds by varying the affected creator population and whether downstream effects are included, which are used for cross-reference without having to rely on strong theoretical assumptions.

\section{Empirical Results}
\label{sec:results}

We evaluate EgoCluster V3 and MultiEgoCluster on LinkedIn's production networks.

\subsection{Experimental Setup}

\paragraph{Data.} We use the LinkedIn feed impression graph (impression\_45):
\begin{itemize}
    \item \textbf{Nodes}: Egos are content creators; alters are content viewers
    \item \textbf{Edges}: Weighted by the number of feed impressions sent by viewers and received by creators over a 45-day lookback window
    \item \textbf{Scale}: At V2/V3 iter0, approximately 2.2 million egos and 52.6 million unique alters in the treatment group
\end{itemize}

\paragraph{Evaluation Period.} Clustering is performed at time T0, with loss rates evaluated at T0 and T14 (14 days later).

\paragraph{Baselines.} We compare against EgoCluster V2 \cite{wesu2024improving} with various loss rate caps as baselines.

\subsection{EgoCluster V3 Performance}

Table~\ref{tab:v3_results} summarizes the performance of EgoCluster V3 compared to V2 across multiple iterations using the impression\_45 graph. This graph contains the impressions members received and sent during the past 45 days. Results are reported for the treatment group.

\begin{table}[t]
\centering
\caption{EgoCluster V3 vs V2 Performance (Impression\_45 Graph, Treatment Group)}
\label{tab:v3_results}
\begin{tabular}{@{}lcccc@{}}
\toprule
\textbf{Version} & \textbf{OLR} & \textbf{Egos} & \textbf{Alters} & \textbf{Ego \%} \\
\midrule
v2\_10pct\_cap & 22.00\% & 156,715 & 850,259 & 2.82\% \\
v2\_20pct\_cap & 28.73\% & 338,577 & 8,012,162 & 6.09\% \\
v2/v3\_iter0 & 30.03\% & 2,209,225 & 52,616,425 & 39.76\% \\
v3\_iter1 & 26.88\% & 2,139,676 & 45,011,734 & 38.51\% \\
v3\_iter2 & 24.38\% & 2,075,941 & 42,811,277 & 37.36\% \\
v3\_iter3 & 22.47\% & 2,014,391 & 40,930,798 & 36.26\% \\
v3\_iter4 & 19.88\% & 1,917,577 & 38,136,138 & 34.51\% \\
v3\_iter5 & 17.51\% & 1,827,816 & 35,694,494 & 32.90\% \\
\bottomrule
\end{tabular}
\end{table}

\begin{figure}[t]
\centering
\includegraphics[width=0.95\linewidth]{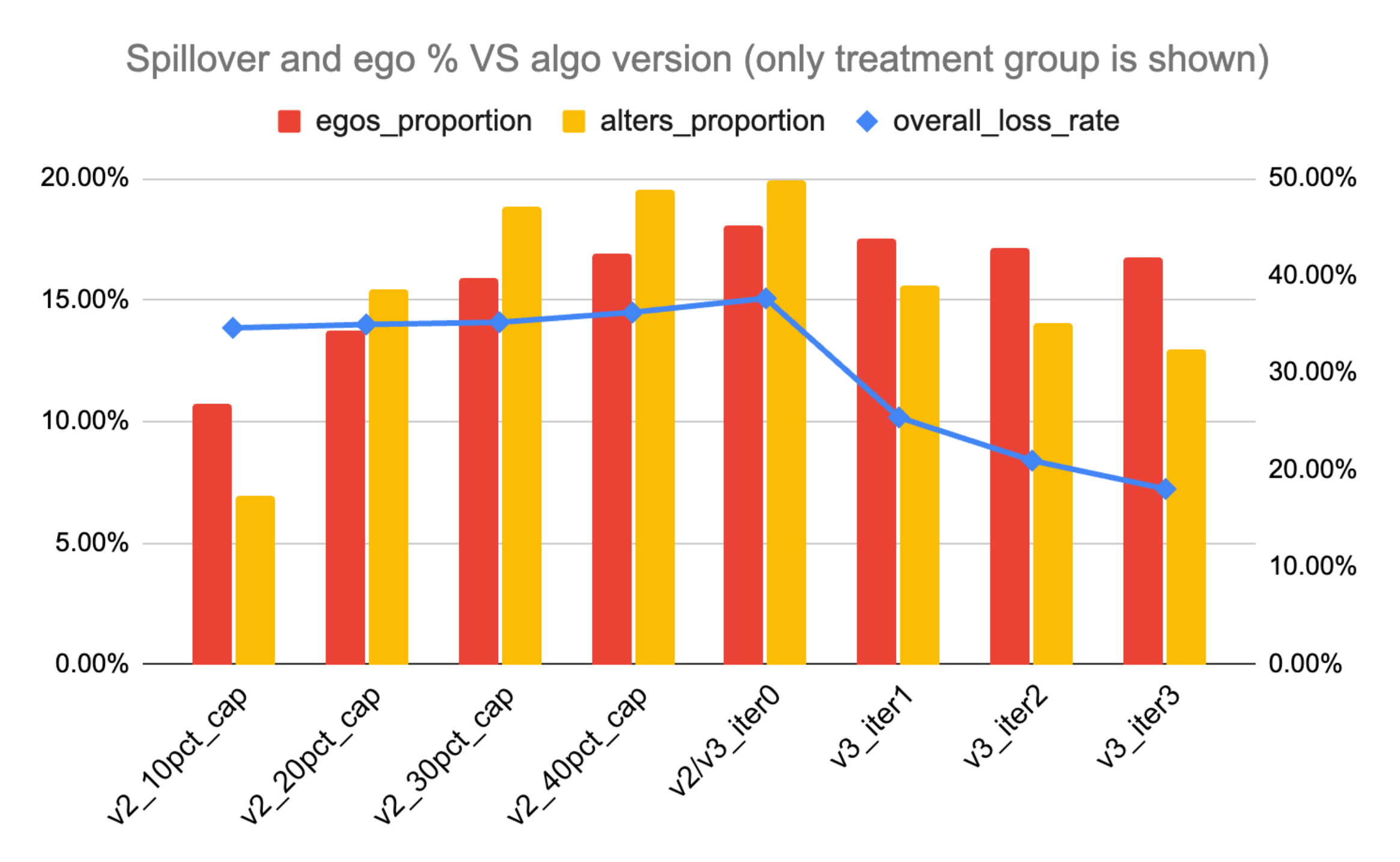}
\caption{\textbf{Loss Rate vs.\ Iteration.}
Loss rate convergence across EgoCluster V3 iterations. The middle bar shows the V2 baseline. Right bars demonstrate that loss rate and sample size both decrease with successive V3 iterations. Left bars show the effect of applying a hard cap on high-loss-rate egos, resulting in reduced sample size with minimal loss rate improvement.}
\Description{Bar chart comparing overall loss rate and sample size across V2 with hard caps, V2/V3 iter0, and V3 iterations 1 through 5. V3 iterations show monotonically decreasing loss rate with modest sample-size reduction.}
\label{fig:loss_convergence}
\end{figure}

\paragraph{Key Findings.}
\begin{itemize}
    \item \textbf{Spillover Reduction}: V3 achieves 42\% reduction in OLR (30.03\% $\rightarrow$ 17.51\%) by iteration 5, while V2 with 20\% cap only achieves 28.73\% OLR with 85\% fewer egos
    \item \textbf{Sample Size Preservation}: V3 iter5 retains 1.83M egos (82.7\% of iter0), compared to V2\_20pct\_cap which retains only 338k egos (15.3\%)
    \item \textbf{Coverage Trade-off}: Alter coverage decreases from 52.6M to 35.7M (68\% retention), still 4.5$\times$ more than V2\_20pct\_cap
    \item \textbf{Power Improvement}: The combination of lower variance (from reduced spillover) and high sample retention yields $\sim$2$\times$ MDE improvement
\end{itemize}

\subsection{MultiEgoCluster Performance}

Table~\ref{tab:multiego_results} compares MultiEgoCluster against V3 on the job marketplace graph.

\begin{table}[t]
\centering
\caption{MultiEgoCluster vs V3 Performance (Job Marketplace)}
\label{tab:multiego_results}
\setlength{\tabcolsep}{3pt}
\begin{tabular}{@{}lccc@{}}
\toprule
\textbf{Metric} & \textbf{V3 Baseline} & \textbf{MultiEgoCluster} & \textbf{Improvement} \\
\midrule
T0 Loss Rate & 15.25\% & 6.65\% & 56.4\% reduction \\
T14 Loss Rate & 18.8\% & 8.45\% & 55.1\% reduction \\
Ego Coverage & 43.5\% & 47.3\% & +8.7\% \\
Alter Coverage & 45.1\% & 62.4\% & +38.4\% \\
\bottomrule
\end{tabular}
\end{table}

\begin{figure}[t]
\centering
\includegraphics[width=0.95\linewidth]{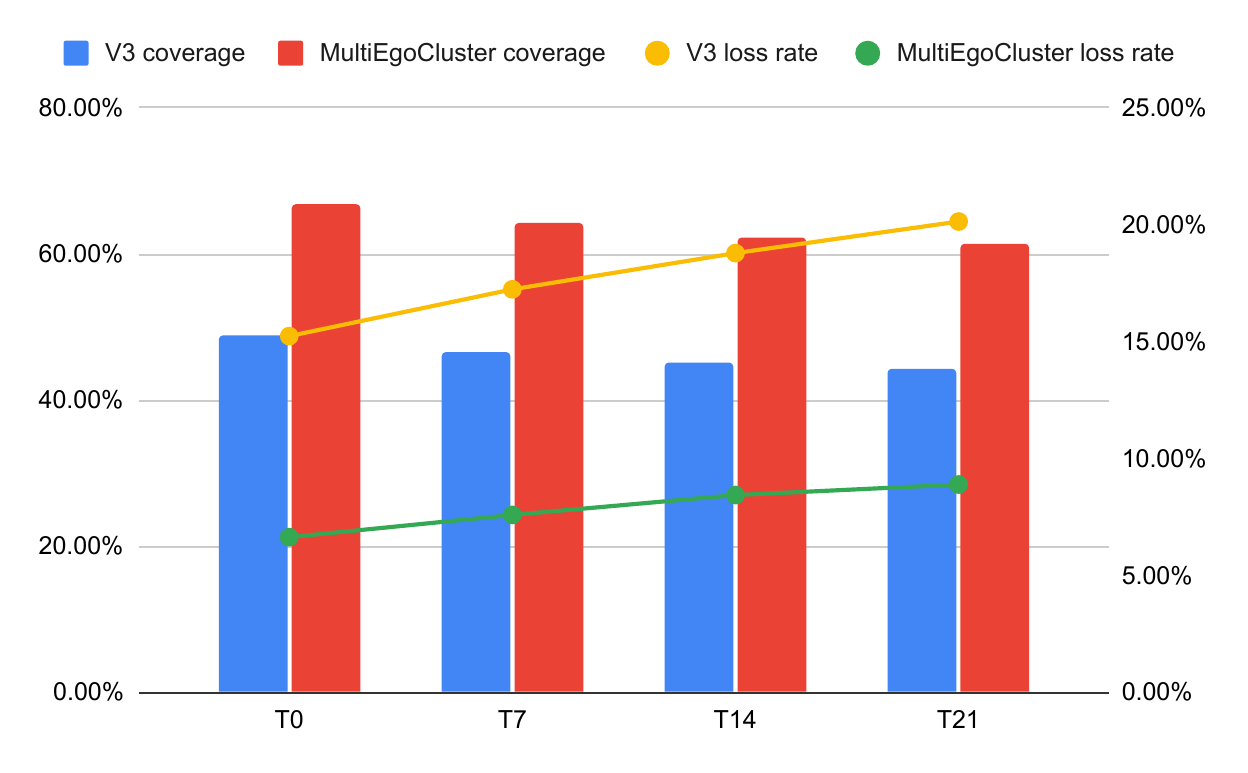}
\caption{\textbf{Spillover and Coverage Comparison of MultiEgoCluster vs.\ EgoCluster V3.}
MultiEgoCluster achieves lower spillover and higher alter coverage across all time periods (T0--T21). It reduces loss rate by a substantial margin and increases alter coverage by over 35 percentage points, enabling larger and more stable samples for experimentation.}
\Description{Two grouped bar charts. The left chart shows loss rate over T0, T7, T14, and T21 for V3 versus MultiEgoCluster, with MultiEgoCluster consistently lower. The right chart shows alter coverage over the same time points, with MultiEgoCluster consistently higher.}
\label{fig:MultiEgoCluster_performance}
\end{figure}

\paragraph{Key Findings.}
\begin{itemize}
    \item \textbf{Additional Spillover Reduction}: MultiEgoCluster reduces spillover by an additional 56\% beyond V3
    \item \textbf{Coverage Increase}: Alter coverage increases by 38\%, as grouping egos prevents the cascade of alter losses when individual egos are filtered
    \item \textbf{Net Power Gain}: The combination of 56\% lower variance and 38\% more alters yields substantial net power improvement
\end{itemize}

\section{Production Impact}
\label{sec:impact}

Both EgoCluster V3 and MultiEgoCluster have been deployed in production at LinkedIn, enabling previously infeasible two-sided experiments.

\subsection{EgoCluster V3: Creator Feedback Experiment}

We applied EgoCluster V3 to measure the creator-side impact of a feed ranking change designed to increase viewers' feedback on creators.

\paragraph{Experiment Design.}
\begin{itemize}
    \item \textbf{Treatment}: New ranking algorithm emphasizing creator feedback to boost creator's engagement
    \item \textbf{Control}: Existing ranking algorithm
    \item \textbf{Duration}: 14 days
    \item \textbf{Clustering}: V3 iter5, achieving $\sim$11\% T0 loss rate
\end{itemize}

\paragraph{Results.}
\begin{itemize}
    \item \textbf{Primary Metrics}: Creator retention and activity metrics showed statistically significant positive impacts
    \item \textbf{Secondary Metrics}: Creator downstream metrics, e.g., their viewers' engagement increase due to notifications.
    \item \textbf{SWI Estimation}: Site-wide impact estimated after CEM reweighting showed consistent positive effects
\end{itemize}

\paragraph{Counterfactual.} With V2's higher loss rate, the predicted MDE would have been insufficient to detect these effects with 80\% power. V3's 2$\times$ power improvement was critical.

\subsection{MultiEgoCluster: Job Marketplace Experiment}

We applied MultiEgoCluster to measure job poster outcomes from a job seeker ranking improvement.

\paragraph{Experiment Design.}
\begin{itemize}
    \item \textbf{Treatment}: Improved job-seeker ranking by relevance
    \item \textbf{Control}: Baseline ranking
    \item \textbf{Duration}: 1 week
    \item \textbf{Clustering}: MultiEgoCluster with $\sim$16\% T0 loss rate
\end{itemize}

This $\sim$16\% production loss rate differs from the 6.65\% offline result in Table~\ref{tab:multiego_results} due to different graph constructions: the production graph applies use-case-specific eligibility filters and a different time window than the offline framework of Section~\ref{sec:results}, though both represent substantial improvements over the $\sim$30\%+ baseline without our methods.

\paragraph{Results.}
\begin{itemize}
    \item \textbf{Primary Metrics}: Job posting performance showed statistically significant positive impacts
    \item \textbf{Secondary Metrics}: Job seeker engagement and quality metrics improved
    \item \textbf{Business Impact}: Enabled launch of ranking improvement with measurable ROI
\end{itemize}

The 38\% increase in sample size and 56\% spillover reduction enabled detection of effects that would have been borderline with V3 alone.

\section{Comparison with Alternative Approaches}
\label{sec:alternatives}

\subsection{Individual Randomization (Full Population)}

Randomizing alters individually maximizes sample size but creates severe spillover. Each ego receives a roughly 50/50 mix of treatment and control alters, eliminating any treatment/control separation at the ego level. For example, in the creator-viewer marketplace, each creator would receive mixed feedback from both treatment and control viewers; similarly, in the job marketplace, each job poster would receive applications from both treatment and control job seekers. Consequently, we would expect to observe 0\% treatment effect on ego metrics in such an experiment, regardless of the true causal impact.

This is not merely a power issue---it is a fundamental identification problem. Without separating egos into treatment-pure and control-pure groups, there is no contrast to estimate.

\subsection{Switchback Designs}

Switchback experiments alternate treatment over time periods rather than across users. Potential advantages include:
\begin{itemize}
    \item No clustering needed; full population included
    \item Temporal separation may reduce some spillover
\end{itemize}

However, for two-sided marketplace experiments, switchback designs face challenges:
\begin{itemize}
    \item \textbf{Carryover Effects}: Ego behavior (e.g., creators posting content, job posters managing listings) is highly autocorrelated. Treatment effects from one period spill into subsequent control periods as egos adapt their behavior.
    \item \textbf{Measurement Window}: Ego outcomes manifest over days, making short switchback periods impractical.
    \item \textbf{Variance}: Temporal correlation inflates variance, resulting in worse MDE than cluster randomization.
\end{itemize}

We conclude that switchback designs are not well-suited for measuring ego-side effects in two-sided marketplace settings.

\subsection{Horvitz-Thompson Estimation}

An alternative is individual randomization with inverse probability weighting:
\begin{equation}
\hat{\tau}^{HT} = \frac{1}{n} \sum_i \frac{Y_i \cdot Z_i}{\pi_i} - \frac{1}{n} \sum_i \frac{Y_i \cdot (1-Z_i)}{1-\pi_i}.
\end{equation}
where $\pi_i$ is the probability that unit $i$ is exposed to treatment (accounting for network exposure).

While theoretically unbiased under partial interference, HT estimators suffer from:
\begin{itemize}
    \item \textbf{High Variance}: In networks with high-degree nodes (max degree $\sim$100k in our case), exposure probabilities become extreme, inflating variance
    \item \textbf{Empirical Performance}: HT confidence intervals are substantially wider than cluster-based estimates in our experiments
\end{itemize}

\section{Discussion and Limitations}
\label{sec:discussion}

\subsection{Assumptions}

Our methods rely on several assumptions:

\paragraph{One-Hop Interference.} We assume spillover effects propagate only one hop (from treated alters to egos). In highly interconnected networks, multi-hop effects may exist but are not modeled.

\paragraph{Additive Spillover.} The bias correction assumes spillover effects are additive in the number of treated alters. Saturation effects (where additional treated alters have diminishing impact) would require modified corrections.

\paragraph{Network Stability.} We assume the network structure remains stable during the experiment. Our 85\% stability rate validates this for 14-day experiments, but longer experiments may require dynamic reclustering.

\paragraph{Feature-Based Stratification.} SWI estimation assumes that stratification features capture the relevant heterogeneity between experimental and population samples.

\subsection{Generalization}

The EgoCluster framework applies to any bipartite marketplace with observed interaction networks:
\begin{itemize}
    \item Feed ecosystems (creators $\times$ viewers)
    \item Job marketplaces (posters $\times$ seekers)
    \item Advertising (advertisers $\times$ audiences)
    \item E-commerce (sellers $\times$ buyers)
\end{itemize}

Key requirements: (1) bipartite network structure, (2) spillover flowing primarily in one direction, (3) sufficient sample size ($>$50k egos) for stable stratified filtering.

\section{Conclusion}
\label{sec:conclusion}

Causal inference in two-sided networked marketplaces faces the challenge of interference between participants. Cluster randomization experiments control the interference and enables cross-sided impact measurement. However, clustering methods have a fundamental trade-off: reducing spillover requires removing high-interference nodes, but naively filtering these nodes eliminates their connections and reduces sample size without much improvement in spillover. This paper introduces techniques to address this challenge and pushes that trade-off frontier out:

\begin{enumerate}
    \item \textbf{EgoCluster V3}: Iterative clustering that reduces spillover by 3$\times$ (from $\sim$28\% to $\sim$9\%) while maintaining 90\%+ alter coverage, enabling 2$\times$ MDE improvement. This achieves a >80\% reduction in spillover compared to individual-level Bernoulli randomization ($\sim$50\% spillover).

    \item \textbf{MultiEgoCluster}: Two-stage clustering achieving $\sim$6\% (near-negligible) spillover with 38\% increased sample size through ego grouping.

    \item \textbf{Bias Correction and SWI Estimation}: Theoretical framework for correcting residual spillover bias and generalizing from non-representative samples using CEM-weighted CATEs.
\end{enumerate}

We deployed these methods in production at LinkedIn, enabling high-impact experiments that measure two-sided effects which were previously undetectable due to high spillover and low power.

\paragraph{Future Work.}
\begin{itemize}
    \item Building multi-hop interference models to handle highly connected networks.
    \item Designing dynamic reclustering process for long-duration experiments.
    \item Measuring heterogeneous spillover effects with covariate-dependent $\lambda$.
\end{itemize}

\begin{acks}
We thank Di Mo, Min Gong, and James Sorenson for their guidance and support for this project. We also thank our leadership team, especially Jia Ding, Kuo-Ning Huang, Tim Jurka, and Wenjing Zhang for their continued support of this work. We appreciate LinkedIn Flagship Experience and Job Marketplace teams for their close collaboration, implementation support, and valuable feedback, as well as discussions with the experimentation and causal inference community that helped shape this work.
\end{acks}

\bibliographystyle{ACM-Reference-Format}
\balance  
\bibliography{egocluster_v3_updated}

\appendix

\section{Algorithm Pseudocode}
\label{app:algorithm}
See Algorithm \ref{alg:v3} for the pseudocode for EgoCluster V3.

\begin{algorithm}[h]
\caption{EgoCluster V3 Iterative Refinement}
\label{alg:v3}
\begin{algorithmic}[1]
\REQUIRE Bipartite graph $G = (E, A, W)$, max iterations $K$, step size $s$
\ENSURE Clustering assignment $C$, loss summary
\STATE $E' \leftarrow$ FilterByNetworkSize$(E)$
\STATE $\text{egoAssignment} \leftarrow$ RandomAssign$(E', \{T, C\})$
\STATE Checkpoint$(\text{egoAssignment})$
\STATE $C_0 \leftarrow$ V2Cluster$(G, \text{egoAssignment})$
\STATE WriteLossSummary$(C_0, \text{iter}=0)$
\FOR{$k = 1$ to $K$}
    \STATE $\text{egoWeights} \leftarrow$ LoadPreviousIteration$(k-1)$
    \STATE $\text{egoWeights} \leftarrow$ JoinFeatures$(\text{egoWeights})$
    \FOR{each stratum $s$ in (variant, features)}
        \STATE $\text{countToDrop} \leftarrow s \cdot |\text{egos in stratum}|$
        \STATE $\text{worstEgos} \leftarrow$ TopByLoss$(\text{stratum}, \text{countToDrop})$
        \STATE $E' \leftarrow E' \setminus \text{worstEgos}$
    \ENDFOR
    \STATE $\text{egoAssignment}' \leftarrow$ Filter$(\text{egoAssignment}, E')$
    \STATE $C_k \leftarrow$ V2Cluster$(G, \text{egoAssignment}')$
    \STATE WriteLossSummary$(C_k, \text{iter}=k)$
\ENDFOR
\RETURN $C_K$
\end{algorithmic}
\end{algorithm}

\section{Power Analysis}
\label{app:figures}

In this section, we compare the minimal detectable effect (MDE) of a true north metric between EgoCluster V2 and V3. In V2, we disregard egos whose loss rate exceeds a threshold (e.g., 80\%, 50\%, 30\%) and compute the MDE on the remaining population. For EgoCluster V3, we compute the MDE for each iteration (up to 5).
The results show that EgoCluster V3 achieves substantially lower MDE compared to V2 at the same sample size and power level. Notably, even at the final iteration where spillover control is strongest, V3 maintains low MDE due to its significantly larger sample size, enabling both stringent spillover control and strong statistical power.

\begin{figure}[H]
\centering
\includegraphics[width=0.95\linewidth]{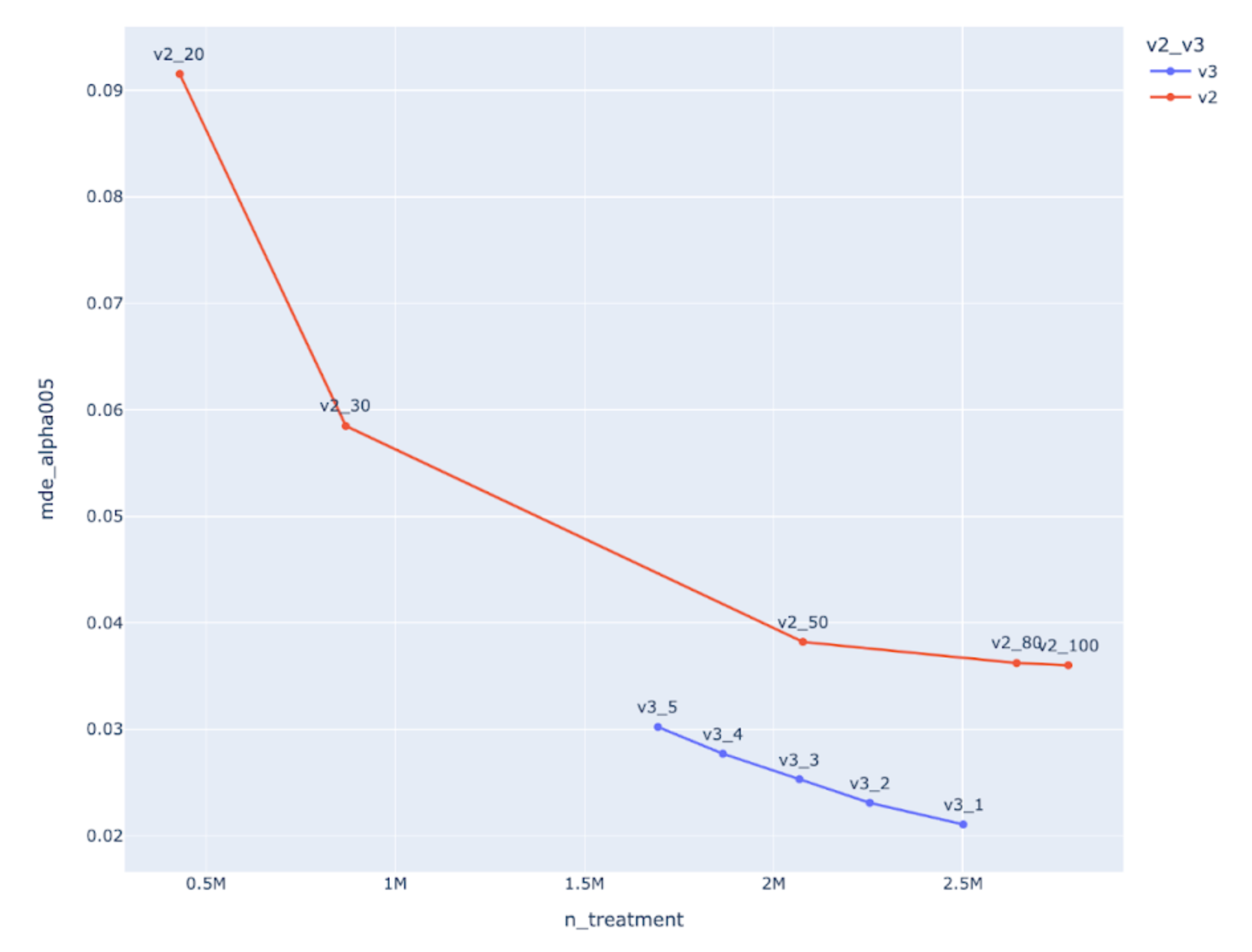}
\caption{\textbf{Minimal Detectable Effect (MDE) comparison.}
V2 applies hard caps on egos with loss rate exceeding X\%; V3 shows MDE across iterations.}
\Description{Line/bar chart comparing minimal detectable effect across EgoCluster V2 with various loss-rate caps versus EgoCluster V3 iterations 1 through 5. V3 achieves uniformly lower MDE.}
\label{fig:mde_power}
\end{figure}

\section{Detailed Proofs}
\label{app:proofs}

\begin{theorem}[Per-Iteration Optimality of Alter Assignment]
\label{thm:alter_optimality}
For any fixed set of active egos $E$ partitioned into treatment and control,
the alter-assignment rule used by EgoCluster V3
within each iteration on the surviving ego set minimizes the overall loss
rate. Consequently, each V3 iteration produces an assignment that is optimal
conditional on the egos remaining after filtering.
\end{theorem}

\begin{proof}
Fix any iteration's surviving ego set $E$ partitioned into $T$ and $C$.
Consider any alter $j$ connected to egos in both variants, and let
$W_T(j) = \sum_{i \in T} w_{ij}$ and $W_C(j) = \sum_{i \in C} w_{ij}$.

The alter-assignment rule assigns $j$ to variant\\
$v^*(j) = \arg\max(W_T(j), W_C(j))$. The loss contribution from $j$ is the
weight to egos in the opposite variant:
\begin{equation}
\text{Loss}(j) = \begin{cases}
W_C(j) & \text{if } v^*(j) = T, \\
W_T(j) & \text{if } v^*(j) = C.
\end{cases}
\end{equation}

By construction, $v^*(j)$ minimizes $\text{Loss}(j)$ for each alter
independently. Because the overall loss is the sum $\sum_j \text{Loss}(j)$ of
per-alter losses, and each term is minimized by an independent choice, the
alter-assignment rule attains the minimum overall loss rate for the given
ego set $E$. Since EgoCluster V3 applies this rule on the surviving ego set
in every iteration, each iteration's clustering is optimal conditional on
the egos that remain after filtering.
\end{proof}

\end{document}